\documentclass[a4paper,UKenglish,cleveref, autoref, thm-restate]{lipics-v2021}

\usepackage[ruled,noend,linesnumbered]{algorithm2e}
\SetKwProg{Def}{define}{:}{}
\SetKw{Next}{next}
\SetKwData{zerosFound}{zerosFound}
\SetKwData{multCount}{multCount}
\usepackage{csquotes}
\usepackage{enumerate}
\usepackage{amsthm}
\usepackage{amsmath}
\usepackage{amsfonts}
\usepackage{amssymb}
\usepackage{faktor}
\usepackage{bm}
\usepackage{booktabs}

\usepackage{graphicx}
\usepackage[clock]{ifsym}
\usepackage{subcaption}

\newtheorem{problem}[definition]{Problem}

\usepackage[colorinlistoftodos,prependcaption,obeyFinal]{todonotes}

\title{On the \texorpdfstring{$p$}{p}-adic Skolem Problem}

\author{Piotr {Bacik}}{University of Oxford, UK \and Max Planck Institute for Software Systems, Saarland Informatics Campus, Germany}{piotr.bacik@stcatz.ox.ac.uk}{https://orcid.org/0009-0006-0248-3204}{Supported by EPSRC grant EP/X033813/1.}

\author{Jo\"el {Ouaknine}}{Max Planck Institute for Software Systems, Saarland Informatics Campus, Germany}{joel@mpi-sws.org}{https://orcid.org/0000-0003-0031-9356}{Supported by ERC grant DynAMiCs (101167561) and DFG grant 389792660 as part of \href{https://perspicuous-computing.science}{TRR~248}.}

\author{David {Purser}}{University of Liverpool, UK}{D.Purser@liverpool.ac.uk}{https://orcid.org/0000-0003-0394-1634}{}

\author{James {Worrell}}{University of Oxford, UK}{jbw@cs.ox.ac.uk}{https://orcid.org/0000-0001-8151-2443}{Supported by EPSRC grant EP/X033813/1.}

\authorrunning{P. Bacik, J. Ouaknine, D. Purser, and J. Worrell}

\Copyright{Piotr Bacik,  Jo\"el Ouaknine, David Purser, and James Worrell}

\nolinenumbers



\acknowledgements{
The authors would like to thank Naoki Fujita for helpful comments on the main algorithm, as well as pointing out an oversight in the previous version of the paper. Jo\"{e}l Ouaknine is also affiliated with Keble College, Oxford as an \href{https://emmy.network}{emmy.network} Fellow.
}

\category{}

\supplement{}
\supplementdetails[subcategory={Online tool}, cite={}, swhid={}]{InteractiveResource}{https://skolem.mpi-sws.org/?padic}
\supplementdetails[subcategory={Source code}, cite={bacik_2025_16794130}, swhid={}]{Software}{https://doi.org/10.5281/zenodo.15858619}

\EventEditors{Meena Mahajan, Florin Manea, Annabelle McIver, and Nguy\~{\^{e}}n Kim Th\'{\u{a}}ng}
\EventNoEds{4}
\EventLongTitle{43rd International Symposium on Theoretical Aspects of Computer Science (STACS 2026)}
\EventShortTitle{STACS 2026}
\EventAcronym{STACS}
\EventYear{2026}
\EventDate{March 9--September 13, 2026}
\EventLocation{Grenoble, France}
\EventLogo{}
\SeriesVolume{364}
\ArticleNo{51}

\newcommand{\Q}{\mathbb{Q}}

\newcommand{\p}{\mathfrak{p}}
\newcommand{\Z}{\mathbb{Z}}
\newcommand{\N}{\mathbb{N}}
\newcommand{\K}{\mathbb{K}}
\renewcommand{\O}{\mathcal{O}}
\newcommand{\C}{\mathbb C}
\newcommand{\LRS}[1]{{\boldsymbol{#1}}}

\counterwithin{equation}{section}

\begin{document}

\maketitle

\begin{abstract}
The Skolem Problem asks to determine whether a given linear recurrence sequence (LRS) has a zero term.  Showing decidability of this problem is equivalent to giving an effective proof of the Skolem-Mahler-Lech Theorem, which asserts that a non-degenerate LRS has finitely many zeros.  The latter result was proven over 90 years ago via an ineffective method
showing that such an LRS has only finitely many $p$-adic zeros.  In this paper we consider the problem of determining whether a given LRS has a $p$-adic zero, as well as the corresponding function problem of computing exact representations of all $p$-adic zeros.  We present algorithms for both problems and report on their implementation. The output of the algorithms is unconditionally correct, and termination is guaranteed subject to the $p$-adic Schanuel Conjecture (a standard number-theoretic hypothesis concerning the $p$-adic exponential function).
While these algorithms do not solve the Skolem Problem, they can be exploited to find natural-number and rational zeros under additional hypotheses.  To illustrate this, we apply our results to show decidability of the \emph{Simultaneous Skolem Problem} (determine whether two coprime linear recurrences have a common natural-number zero), again subject to the $p$-adic Schanuel Conjecture.  \end{abstract}
\newpage
\section{Introduction}
\subsection{The Skolem Problem} \label{sec:intro}
A \emph{linear recurrence sequence} (LRS) $\LRS{u} = \langle u_n \rangle_{n=0}^\infty$ is a sequence of integers
satisfying a {linear recurrence relation}:
\begin{align}\label{eqn:LRS_def}
u_{n+d} = a_{d-1} u_{n+d-1} + \dots + a_0 u_n
\end{align}
where $a_0, \dots , a_{d-1} \in \Z$ and $a_0\neq 0$.
We call $d$ the order of the recurrence.  If $d$ is the minimum order of a recurrence satisfied by $\LRS{u}$ then we call
$d$ the \emph{order} of $\LRS{u}$. 

The zero set of an LRS $\LRS{u}$ is $\{ n \in \mathbb N : u_n=0\}$.
The celebrated Skolem-Mahler-Lech theorem~\cite{Skolem_SML,Mahler_SML,lech_note_1953} states that the zero set is  
comprised of a union of a finite set and finitely many arithmetic progressions. The statement may be refined via the concept of \emph{non-degeneracy}. Define the \emph{characteristic polynomial}\footnote{Note that the minimal-order recurrence for an LRS $\LRS{u}$ is unique, and so the characteristic polynomial of $\LRS u$ is also unique.} of the recurrence~\eqref{eqn:LRS_def} to be
\begin{align}
    g(X) := X^d - a_{d-1}X^{d-1} - \dots - a_0 \, .
\end{align}
Let $\lambda_1, \dots, \lambda_s \in \mathbb{\overline{\mathbb Q}}$ be the distinct roots of $g$; these are called  the \emph{characteristic roots} of $\LRS{u}$. We say $\LRS{u}$ is \emph{non-degenerate} if no ratio $\lambda_i/\lambda_j$ of distinct characteristic roots is a root of unity.  A given LRS can be effectively decomposed as the interleaving of finitely many non-degenerate LRS~\cite[Theorem 1.6]{recurrence}.  The core of the Skolem-Mahler-Lech Theorem is that a non-degenerate LRS that is not identically zero has finitely many zero terms.
Unfortunately, the proof of this result remains ineffective:
there is no known algorithm to determine whether a non-degenerate LRS has a zero.
This is the famous \emph{Skolem Problem}:
\begin{problem}[The Skolem Problem]
Given an LRS $\LRS{u}$  specified by a non-degenerate linear recurrence and a set of initial values, determine whether there exists $n \in \mathbb N$ such that $u_n=0$.
\end{problem}
One can also formulate a corresponding function version of this problem in which the task is to compute the finite set of zeros of a given non-degenerate LRS\@.  We denote the decision version by $\mathtt{SP}(\mathbb N)$ and the function version
by $\mathtt{FSP}(\mathbb N)$.  This notation makes explicit that we are looking for natural-number zeros of the LRS.  

It is folklore that computability of $\mathtt{FSP}(\mathbb N)$ reduces to decidability of $\mathtt{SP}(\mathbb N)$.  Assuming the latter, given an LRS $\LRS{u}$, the finitely many zeros in each non-degenerate subsequence of $\LRS{u}$ can be found by brute-force search and, since the infinite suffix of an LRS remains an LRS, one can use a decision procedure for $\mathtt{SP}(\mathbb N)$ to certify that no zeros remain to be found.  However, decidability of the Skolem Problem has remained open for close to a century, with only partial results known from restricting the order (see the exposition of~\cite{bilu2025} on results of \cite{Mignotte_distance,Vereshchagin,bacik_completing_2025}), restricting to reversible sequences of low order~\cite{Lipton_2022}, or restricting to \emph{simple} LRS and assuming certain number-theoretic conjectures~\cite{SkolemMeetsSchanuel}.

In the remainder of this section we introduce various relaxations of the Skolem Problem that arise by extending LRS to larger domains and seeking zeros of such extensions.
\subsection{The Bi-Skolem Problem}
The first variant of the Skolem Problem involves bi-infinite (that is, two-way infinite) sequences.
Indeed, given a recurrence~\eqref{eqn:LRS_def} and initial values $u_0,\ldots,u_{d-1}\in \Q$, there is a unique bi-infinite sequence
$\LRS{u} = \langle u_n \rangle_{n = -\infty}^\infty$ that satisfies the recurrence.  We call $\LRS{u}$ 
a \emph{linear recurrent bi-sequence} (LRBS)\@.
For example, the Fibonacci sequence extends to an LBRS
$\langle \ldots,5,-3,2, -1,1,0,1,1,2,3,5\ldots\rangle$.
\begin{problem}[Bi-Skolem Problem]
  Given an LRBS $\LRS{u}$,  specified by a non-degenerate linear recurrence and a set of initial values, determine whether there exists $n \in \mathbb Z$ such that $u_n=0$.
\end{problem}
We use the notation $\mathtt{SP}(\mathbb Z)$ to refer to the above decision problem and we write
$\mathtt{FSP}(\mathbb Z)$ for the corresponding function version, in which we output the finite set of zeros of a given non-degenerate bi-infinite sequence.

For function problems $P_1, P_2$, write $P_1 \leq P_2$ if there is a Turing reduction from $P_1$ to $P_2$. Here for purposes of comparison we view decision problems as function problems with outputs in $\{\texttt{TRUE},\texttt{FALSE}\}$. Write $P_1 \equiv P_2$ if $P_1 \leq P_2$ and $P_2 \leq P_1$. Then it is easy to see that
  \[ \mathtt{SP}(\mathbb Z) \leq \mathtt{FSP}(\mathbb Z) \equiv \mathtt{FSP}(\mathbb N) \equiv \mathtt{SP}(\mathbb N)
  \, .\]
  Indeed, the reduction $\mathtt{FSP}(\mathbb Z) \leq \mathtt{FSP}(\mathbb N)$ is realised by splitting a given bi-infinite LRS around the index zero  into forward and backward sequences
  (both of which are LRS) and computing the respective zeros of each of the two one-way infinite sequences.

  Unlike for Skolem's Problem, it is not known whether the function version of the Bi-Skolem Problem reduces to its decision version, i.e., it is not known whether $\mathtt{FSP}(\mathbb Z) \leq \mathtt{SP}(\mathbb Z)$; however the reduction holds if one assumes
  the weak $p$-adic Schanuel Conjecture~\cite{SkolemMeetsSchanuel}. 

  \subsection{The Rational Skolem Problem}
  Having expanded the index set of an LRS to $\Z$ in the Bi-Skolem Problem, we consider a further expansion of its domain to
  $\Q$, which leads us to consider \emph{rational zeros} of an LRS\@.
  One way to realise this generalisation is via the exponential-polynomial formulation of LRS\@.
  It is classical that an LRS $\boldsymbol{u}$ of order $d$ with characteristic roots $\lambda_1, \dots , \lambda_s$ admits the following representation:
  \begin{align}\label{eqn:Binet}
    u_n = \sum_{i=1}^s q_i(n)\lambda_i^n \, ,
\end{align}
where the $q_i$ are polynomials with algebraic coefficients and degree one less than the multiplicity of $\lambda_i$ as a root of the characteristic polynomial of $\LRS{u}$.
We say that $\frac{a}{b} \in \mathbb Q$ is a rational zero of $\LRS{u}$ if $\sum_{i=1}^s q_i(\frac{a}{b}) \lambda_i^{\frac{a}{b}}=0$,
where $\lambda^{\frac{a}{b}}$ denotes any $b$-th root of $\lambda^a$.  For example, the sequence
$u_n = 4^n+2$ has a rational zero at $\frac{1}{2}$ that is witnessed by setting $4^{1/2}:=-2$.
The Rational Skolem Problem $\mathtt{SP}(\mathbb Q)$ asks to determine whether a given LRS has a rational zero, while its function analogue
$\mathtt{FSP}(\mathbb Q)$ asks to compute all rational zeros of a non-degenerate sequence.
We note that 
the definition of rational zeros is consistent with that of integer zeros, that is, the integer rational zeros of $\boldsymbol u$ are precisely the zeros of the bi-infinite extension of $\boldsymbol u$.
  
A recent result~\cite[Proposition 4.7]{bilu_twisted_2025} shows that the denominator of any rational zero can be effectively bounded. This allows us to determine the relationship of the Rational Skolem Problem with the Bi-Skolem Problem and the usual Skolem Problem.  We have
\begin{align*}
\mathtt{SP}(\mathbb Q) \overset{(1)}{\equiv}
\mathtt{SP}(\mathbb Z) \overset{(2)}{\leq}  \mathtt{SP}(\mathbb N) \equiv \mathtt{FSP}(\mathbb N)
\equiv \mathtt{FSP}(\mathbb Z) \overset{(3)}{\equiv} \mathtt{FSP}(\mathbb Q) \, ,
  \end{align*}
  where, as noted earlier, $(2)$ is an equivalence assuming the weak $p$-adic Schanuel Conjecture. Indeed, suppose $\LRS{u}$ is an LRS satisfying~\eqref{eqn:Binet}. We may compute a bound $b_\text{max}$ on the largest denominator of any rational zero of $\LRS{u}$ using~\cite[Proposition 4.7]{bilu_twisted_2025}.  For the reduction (1) we note that deciding $\mathtt{SP}(\mathbb Q)$
  is equivalent to deciding $\mathtt{SP}(\mathbb Z)$ on every LRS $\LRS{v}$ defined by
$v_n:=  \sum_{i=1}^s  \textstyle q_i(\frac{n}{b}) (\lambda^{1/b}_i)^n$, where
$1\leq b \leq b_\text{max}$.  The reduction (3) is similar.
\subsection{The \texorpdfstring{$p$}{p}-adic Skolem Problem}
The main contributions of this paper concern the $p$-adic zeros of an  LRS, that is, zeros lying in a $p$-adic completion $\mathbb Z_p$ of the integers with respect to a given prime $p$.
Roughly speaking, the idea is to extend a given LRS $\LRS{u} = \langle u_n\rangle_{n=0}^\infty$ to a map $f:\mathbb Z_p\rightarrow \mathbb Z_p$
such that $f(n)=u_n$ for all $n\in \mathbb Z$.  We then determine whether $f$ has any zeros in $\Z_p$ and, if yes, we compute exact representations and approximate them to arbitrary precision. By an exact representation of a $p$-adic zero $x$ of $f$ we mean a disc $D$ in $\Z_p$ such that $x$ is the unique zero of $f$ in $D$.\footnote{This notion may be compared
to exact representations of algebraic numbers given by their minimal polynomial and an approximation that distinguishes the number from other roots of the minimal polynomial.}
It turns out that $f$ can have $p$-adic zeros other than the integer or rational zeros of the original LRS $\LRS{u}$.  Hence the ability to determine the existence 
of $p$-adic zeros does not directly solve Skolem's Problem.
Nevertheless, as our results and experiments show, working $p$-adically offers a practical approach to finding integer zeros of an LRS that can moreover find all integer zeros subject to additional assumptions and hypotheses.

By way of example, consider the ring $\mathbb Z_3$ of 3-adic integers,
which is the Cauchy completion of $\mathbb Z$ with respect to the absolute value $|\cdot|_3$.  The latter is defined by writing
$|a|_3:=3^{-k}$, where $k$ is the order of 3 as a divisor of $a\in \mathbb Z$ (the larger $k$, the smaller the absolute value).
It turns out that the LRS $u_n=4^n+2$ extends uniquely to a continuous map $f:\mathbb Z_3\rightarrow \mathbb Z_3$.  It is not difficult to see that $f(\frac{1}{2})=0$.
Indeed, write $n_k:=\frac{1+3^k}{2}$ for all $k\in \mathbb N$.  Then
$\lim_{k\rightarrow\infty} n_k=\frac{1}{2}$  in $\mathbb Z_3$.
On the other hand, it can be shown by induction that $u_{n_k} \equiv 0 \bmod {3^{k+1}}$ for all $k$ and hence $\lim_{k\rightarrow\infty}u_{n_k}=0$ in $\mathbb Z_3$.
By continuity we conclude that $f(\frac{1}{2})=0$.

In general, given an LRS $\LRS u$ satisfying~\eqref{eqn:LRS_def}, pick a prime $p \in \N$ such that $p$ does not divide the last coefficient $a_0$ of the recurrence~\eqref{eqn:LRS_def}. Then there exists $N \geq 1$ and analytic functions $f_0,\ldots,f_{N-1} : \Z_p \to
\Z_p$ such that for each $0 \leq \ell \leq N-1$ we have $u_{Nn+\ell} = f_\ell(n)$ for all $n\in \mathbb N$.
Let us call a zero $x \in \Z_p$ of one of the above functions $f_\ell$ a \emph{$p$-adic zero of $\LRS{u}$}.
The proof of the Skolem-Mahler-Lech Theorem shows that $\LRS{u}$ has finitely many $p$-adic zeros---but it does not allow one to determine 
whether a given $\LRS{u}$ has any
$p$-adic zeros and, if so, how to compute them (hence the proof does not tell us anything about computing the integer zeros either).  This motivates:
\begin{problem}[The $p$-adic Skolem Problem]
Given a non-degenerate LRS $\LRS{u}$ and prime $p \in \N$ not dividing the last coefficient of the recurrence satisfied by $\LRS{u}$, determine whether $\LRS{u}$ has a $p$-adic zero. 
\end{problem}
The function version of this problem asks to compute a finite representation of all $p$-adic zeros of $\LRS{u}$ (see Definition~\ref{def:spec}), allowing us both to count the number of $p$-adic zeros and to approximate them to arbitrary precision (with respect to the $p$-adic absolute value). 
As shorthand we denote the decision and function problems respectively by $\mathtt{SP}(\Z_p)$ and $\mathtt{FSP}(\Z_p)$. 

We remark that $p$-adic zeros have a concrete interpretation requiring little $p$-adic terminology to understand. A $p$-adic zero can be interpreted as a coherent collection of zeros mod $p^r$ for all $r$, that is, $x \in \Z_p$ is a $p$-adic zero of $\LRS{u}$ if and only if there is an infinite sequence $n_1,n_2, \dots \in \N$ such that $u_{n_r} \equiv 0 \bmod p^r$ and $n_r \equiv n_{r+1} \equiv x \bmod p^r$. We then have $x \in \N$ if this sequence can be taken eventually constant.

It is open whether any of the above-mentioned variants of Skolem's Problem can be reduced to 
$\texttt{SP}(\Z_p)$ or $\texttt{FSP}(\Z_p)$.
The essential problem is that we do not know how to determine in general whether a $p$-adic zero of a given LRS is rational or not.
While our decision procedure can approximate such a zero to arbitrary precision, it cannot in general certify that it is irrational or even non-integer. Note also that not all rational zeros will be $p$-adic zeros for a given prime $p$.
See Sections~\ref{sec:SkolemConjecture} and~\ref{sec:rat_zeros} for further discussion on these points.

\subsection{Main Results}
Our main theoretical result is the following.
\begin{theorem} \label{thm:pDecidable}
Assuming the $p$-adic Schanuel Conjecture, $\mathtt{SP}(\Z_p)$ is decidable and $\mathtt{FSP}(\Z_p)$ is computable.
\end{theorem}
In the case when all characteristic roots of $\LRS{u}$ lie in $\Z_p$ (which occurs for infinitely many primes $p$ by the Chebotarev density theorem), we have implemented the algorithm for $\texttt{FSP}(\Z_p)$ in the \textsc{Skolem} tool \cite{bacik_2025_16794130}.\footnote{\textsc{Skolem} may be experimented with online at \url{https://skolem.mpi-sws.org/?padic}. For the algorithm described in this paper, toggle the switch labelled ``Use p-adic algorithm''.} 

The main technical lemma (also subject to the $p$-adic Schanuel Conjecture) behind the proof of Thm.~\ref{thm:pDecidable} shows that if two exponential polynomials are coprime in the ring of exponential polynomials, then every common $p$-adic zero must be rational (and thus may be found by enumerating the rationals). This lemma has consequences for the \emph{Simultaneous Skolem Problem}: determine whether two LRS have a common integer zero. 
\begin{restatable}{theorem}{SIMP}
Assuming the $p$-adic Schanuel Conjecture, the Simultaneous Skolem Problem is decidable for coprime LRS.
\label{thm:simultaneous}
\end{restatable}


An in-principle decision procedure for $\texttt{SP}(\Z_p)$ may be deduced from the doctoral thesis of Mariaule \cite{Mariaule} that shows the first-order theory of the structure $(\Z_p,+,\cdot,0,1,E_p)$ is decidable assuming the $p$-adic Schanuel Conjecture (where $E_p(x) = e^{px}$). The $p$-adic Schanuel Conjecture is used in \cite{Mariaule} via a desingularisation construction; given a zero $x$ of an exponential polynomial, there exists a system of exponential polynomials of which $x$ is a non-singular zero. This approach is not practical to implement however, as the algorithm entails enumerating systems of integer multivariate polynomials. In contrast, our algorithm only requires enumerating rational numbers to check for repeated zeros of a given LRS\@.  Such a search can moreover be done efficiently in practice
by searching for rational numbers close to the approximations of $p$-adic zeros that the algorithm identifies.
Using $p$-adic approximation to speed up the search for integer zeros of LRS is explored further in~\cite{bacik_complexity_2025}. 

We note also that Thm.~\ref{thm:simultaneous} does not follow directly from the result of~\cite{Mariaule} since $\mathbb Z$ is not first-order definable in  $(\Z_p,+,\cdot,0,1,E_p)$, as this would immediately contradict decidability of the latter structure. 

\section{Preliminaries} 
\label{sec:prelims}
\subsection{\texorpdfstring{$p$}{p}-adic numbers}
We briefly recall relevant notions about $p$-adic numbers. More details may be found in algebraic number theory textbooks such as \cite{neukirch_algebraic_1999,robert_course_2000}.

Given a prime number $p$, every non-zero rational number $x \in \Q$ may be written as $x = \frac{a}{b}p^r$ for some integers $a,b$ coprime to $p$, with $b$ non-zero, and $r \in \Z$. We define the valuation $v_p(x) = r$, and $v_p(0) = \infty$.  From this derivation we see that $v_p$ satisfies the ultrametric inequality
$v_p(x+y) \geq \min(v_p(x),v_p(y))$ for all $x,y \in \mathbb Q$.

We define an absolute value on $\Q$ by $|x|_p = p^{-v_p(x)}$.  The ultrametric inequality on $v_p$ translates to the 
strong triangle inequality: $|x+y|_p \leq \max(|x|_p,|y|_p)$ for all $x,y \in \mathbb Q$.
Define the set of $p$-adic numbers $\Q_p$ as the completion of $\Q$ with respect to $|\cdot|_p$. If one completes $\Z$ with respect to $|\cdot|_p$ then one gets the $p$-adic integers $\Z_p$, which may also be defined as the unit disc in $\Q_p$:
\begin{align*}
\Z_p = \{x \in \Q_p : |x|_p \leq 1\} = \{x \in \Q_p : v_p(x) \geq 0 \} \, .
\end{align*}
The absolute value $|\cdot|_p$ on $\mathbb Q_p$ extends uniquely to an absolute value on the algebraic closure $\overline{\mathbb Q_p}$.  Specifically, given $\alpha \in \overline{\mathbb Q_p}$, we define $|\alpha|_p := |N(\alpha)|_p^{1/n}$, where $n$ is the degree of $\alpha$ over $\mathbb Q_p$ and $N(\alpha)$ is the \emph{norm} of $\alpha$, that is the determinant of the $\mathbb Q$-linear transformation 
$\mu_\alpha : \mathbb Q(\alpha)\rightarrow \mathbb Q(\alpha)$ given by $\mu_\alpha(x)=\alpha x$.
We correspondingly extend the valuation $v_p(\cdot)$ to $\overline{\mathbb Q_p}$ by $v_p(\alpha) := \frac{1}{n}v_p(N(\alpha))$.  With these definitions the identity 
$|\alpha|_p=p^{-v_p(\alpha)}$ holds for all $\alpha \in \overline{\mathbb Q_p}$.  While $\overline{\mathbb Q_p}$ is not complete with respect to $|\cdot|_p$, taking the Cauchy completion we obtain a field $\mathbb C_p$ that is both algebraically closed and complete under the extension of $|\cdot|_p$ to $\mathbb C_p$. Denote the valuation ring of $\C_p$ by $\O_{\C_p} = \{x \in \C_p: v_p(x) \geq 0\}$. 
\subsection{Computing with \texorpdfstring{$p$}{p}-adic numbers}
We will require a generalisation of the above discussion to deal with algebraic numbers lying outside of $\Z_p$. Given a number field $\K$, let $\O$ denote its ring of integers. Previous ideas of factorising to define a valuation $v_p$ do not carry over directly as $\O$ is no longer a unique factorisation domain, instead one generalises to consider factorisations of (fractional) ideals. Define a \emph{fractional ideal} of $\K$ to be a non-zero finitely generated $\O$-submodule of $\K$; equivalently, $I$ is a fractional ideal if and only if there is $c \in \K$ such that $cI \subseteq \O$ is an ideal of $\O$. Any fractional ideal has a unique factorisation 
\begin{align*}
    I = \prod_{i=1}^t \p_i^{n_i}
\end{align*}
where $t \in \N$ and each $\p \subseteq \O$ is a prime ideal, $n_i \in \Z$ \cite[p. 22]{neukirch_algebraic_1999}. If $I \subseteq \O$ is an ideal then $n_i \geq 0$. Define the valuation $v_\p$ by $v_\p(a) = n$ if the exponent of $\p$ in the prime factorisation of the ideal $a\O$ is $n$, and $v_\p(0) = \infty$. Any prime ideal $\p$ has $v_\p(p) > 0$ for a unique integer prime $p \in \N$. We say $\p$ ``divides'' or ``lies above'' $p$. Say $p$ \emph{ramifies} in $\K$ if $v_\p(p) =e > 1$; call $e = e_{\p/p}$ the \emph{ramification index}. Define the \emph{residue field degree $f = f_{\p/p}$} by $f = \left[\faktor{\O}{\p}:\faktor{\Z}{p\Z}\right]$. Define the $\p$-adic absolute value by $|x|_\p = p^{-fv_\p(x)}$.\footnote{There are many conventions for the normalization of $\p$-adic absolute values as they are equivalent and induce the same completions.} Denote the completion of $\K$ with respect to $|\cdot|_\p$ by $\K_\p$, which may be embedded into $\C_p$. Let the valuation ring be denoted by 
\begin{align*}
    \O_\p = \{x \in \K_\p: v_\p(x) \geq 0\}\, .
\end{align*}
A \emph{uniformiser} $\pi \in \O_\p$ is an element such that $v_\p(\pi) = 1$. For example, let $\K = \Q(\sqrt3)$ and $\p = (\sqrt{3})\O \subseteq \O$. Then $\pi = \sqrt{3}$ is a uniformiser in $\K_\p$, and the ramification index $e_{\p/p} = 2$. Informally, the uniformiser $\pi$ takes the analogous role in $\O_\p$ that $p$ does in $\Z_p$.

Any element $x \in \O_\p$ can be represented as an infinite series $x = \sum_{n=0}^\infty a_n\pi^n$, where each $a_n \in A$, for $A$ a choice of representatives for the residue field $\faktor{\O}{\p}$. In the algorithm, using this fact, $\pi$ and $A$ are used to enumerate elements $x \in \O_\p$, and to easily compute their valuations $v_\p$, since $v_\p(\sum_{n=0}^\infty a_n \pi^n)$ is equal to $j$, where $j$ is the smallest integer for which $a_j \neq 0$. Practically, since only finitely many digits may be stored, for such computations one works with approximations --- that is, reductions mod $\pi^{r}$. This is equivalent to truncating the series representation, i.e. $\sum_{n=0}^\infty a_n \pi^n \equiv \sum_{n=0}^{r-1} a_n \pi^n \bmod \pi^r$.

Since $\faktor{\O}{\p}$ is a finite field, $\faktor{\O}{\p} \cong \mathbb F_{p^f}$ is generated by (the reduction modulo $\p$ of) a single element $a \in \O$; therefore $A$ may be taken to be of the form 
\begin{align} \label{eqn:reps}
    A = \{b_0 + b_1 a + \dots + b_{f-1}a^{f-1} : b_i \in \{0,1,\dots,p-1\}\} \, .
\end{align}
We will not expand much in this article on the intricacies of computations in $\O_\p$, for us it suffices to, given a number field $\K$ and integer prime $p$, pick a prime ideal $\p \subseteq \O$ lying above $p$, compute a uniformiser $\pi \in \O_\p$, and $A$ of the form \eqref{eqn:reps}. This can be handled using Montes' algorithm, see \cite{nart_okutsu-montes_2011}. In particular one may compute representations of all prime ideals lying above $\p$, compute valuations $v_\p$, and a $p$-integral basis of $\K$, from which one can find $A,e_{\p/p},f_{\p/p}$; a uniformiser $\pi$ can simply be found by enumerating elements until one is found with $v_\p(\pi)=1$.


\subsection{\texorpdfstring{$p$}{p}-adic interpolation of linear recurrence sequences} \label{sec:interpolation}
Consider an LRS $\LRS{u}$ satisfying \eqref{eqn:LRS_def} that is given by the formula~\eqref{eqn:Binet} and let $\mathbb K$ be the field generated by the characteristic roots of $\LRS{u}$. As mentioned in Section \ref{sec:intro}, by \cite[Theorem 1.6]{recurrence} we may effectively decompose $\LRS{u}$ into non-degenerate or identically zero subsequences; we may thus assume in the rest of this article that any given LRS $\LRS{u}$ is non-degenerate.

Pick a prime ideal $\p \subseteq \mathcal O$ lying above prime $p \in \N$ that does not divide the last coefficient $a_0$ of the recurrence for $\LRS{u}$ (so that  
 $v_\p(\lambda_i) = 0$ for all $i$). Choose $N \in \Z_{\geq 1}$ to be the smallest positive integer such that $v_\p(\lambda_i^N-1)>\frac{e}{p-1}$ for all $i$, where $e=e_{\p/p}$ is the ramification index.
 For $\ell \in \{0,1,\ldots,N-1\}$ we 
 define the $\ell$-th $\p$-adic interpolant of $\LRS{u}$ to be the analytic function\footnote{The $p$-adic analytic functions $\exp,\, \log$ have the usual series definitions, see \cite[section 5.4]{robert_course_2000}. In particular, $\log$ maps $1 + \pi^r \O_\p$ to $\pi^r \O_\p$, and $\exp$ converges on and maps $\pi^r \O_\p$ to $1 + \pi^r \O_\p$, when $r > \frac{e_{\p/p}}{p-1}$.}
 \begin{align}\label{eqn:interpolate}
     f_\ell(x) = \sum_{i=1}^s q_i(Nx+\ell) \lambda_i^\ell \exp(x \log \lambda_i^N)
 \end{align}
 with $x \in \O_{\p}$. Then  $f_\ell(n) = u_{Nn+\ell}$ for each $n \in \Z_{\geq 0}$
 and so~\eqref{eqn:interpolate} defines an extension of $\LRS{u}$ to $\mathcal O_{\p}$.
 In fact, the right-hand side of~\eqref{eqn:interpolate} converges for any $x \in \O_{\C_p}$ and we call
 such an $x$ an \emph{extended $p$-adic zero} if $f_\ell(x)=0$ for some $\ell\in\{0,\ldots,N-1\}$. Note that since $f_\ell(n) \in \Z$ for all $n \in \N$, by continuity $f_\ell(x) \in \Z_p$ for all $x \in \Z_p$. Thus we have a power series representation
 \begin{align*}
     f_\ell(x) = \sum_{j=0}^\infty b_j x^j
 \end{align*}
for $b_j \in \Z_p$. We will often implicitly restrict the domain to $\Z_p$ and write that $f_\ell : \Z_p \to \Z_p$ is a $p$-adic interpolant of $\LRS{u}$.

Given characteristic roots $\lambda_1,\dots,\lambda_s \in \K$, say they are multiplicatively dependent if there is an integer tuple $(k_1,\dots,k_s) \in \Z^s$ such that $\lambda_1^{k_1}\dots \lambda_s^{k_s}= 1$. The group of all multiplicative relations has generators which are effectively bounded by Masser's theorem \cite{Masser_1988}, this may be used to find all multiplicative relations (see also \cite{combot_computing_2025} for a polynomial-time algorithm) and hence find a multiplicatively independent set of algebraic numbers $\{\alpha_1,\dots,\alpha_r\}$ such that each $\lambda_i$ may be expressed as a Laurent monomial in $\alpha_1,\dots,\alpha_r$ (they can be expressed in this way by non-degeneracy). Therefore $f_\ell$ may be expressed as a Laurent polynomial with algebraic coefficients
\begin{align*}
    f_\ell(x) = P_\ell(x,\exp(x\log \alpha_1^N),\dots,\exp(x \log \alpha_r^N))\, .
\end{align*}
$P_\ell$ may easily be factored as $P_\ell = R_\ell Q_{\ell,1}^{t_1}\dots Q_{\ell,m}^{t_m}$, where $R_\ell$ is a Laurent monomial in $\exp(x\log \alpha_1^N),\dots,\exp(x \log \alpha_r^N)$ and $Q_{\ell,1},\dots,Q_{\ell,m}$ are irreducible polynomials in $x$, $\exp(x \log \alpha_1^N),\dots,\exp(x \log \alpha_r^N)$. Since $R_\ell$ is never zero, finding zeros of $f_\ell(x)$ reduces to working on each irreducible factor $Q_{\ell,i}$ individually. Note also that $R_\ell$ corresponds to an order-1 LRS; in effect we are factorising $u_{Nn+\ell}$ as a product of irreducible (algebraic-valued) LRS. 

We must still explain how $p$-adic zeros are represented and stored. The following is a key lemma for finding zeros of power series on $\K_\p$ (a proof may be found in e.g. \cite[Thm. 9.2]{conrad_hensel} or \cite[Thm. 27.6]{Schikhof_1985}).
\begin{theorem}[Hensel's Lemma for power series] \label{thm:hensel}
  Let $f$ be a power series with coefficients in $\O_{\p}$ that converges on $\O_{\p}$. If $a \in \O_{\p}$ satisfies 
\begin{align*}
|f(a)|_\p < |f'(a)|_\p^2 
\end{align*}
then there is a unique $\alpha \in \O_{\p}$ such that $f(\alpha) = 0$ and $|\alpha-a|_\p < |f'(\alpha)|_\p$. 
\end{theorem}
The following definition explains how we treat $p$-adic zeros of power series:
\begin{definition}
  A specification of an extended $p$-adic zero $x \in \O_\p$ of an analytic function $f: \O_{\p} \to \O_{\p}$ 
    is an integer $y \in \Z$, such that $|f(y)|_\p < |f'(y)|_\p^2$, and $|x-y| < |f'(y)|_\p$.
\label{def:spec}
\end{definition}
Hensel's Lemma ensures that $x$ is defined uniquely. The proof of Hensel's lemma also provides an efficient method for computing approximations of $x$, via the iteration $x_{n+1} = x_n-f(x_n)/f'(x_n)$ and $x_0 = y$.

The next theorem is a generalisation of Strassman's Theorem, commonly known via the notion of the Newton Polygon; see \cite{robert_course_2000}, page 307.
\begin{theorem}\label{thm:newton_polygon}
Let $p$ be a prime, and let $f(X) = \sum a_n X^n \in \O_{\C_p}[[X]]$ be a nonzero convergent power series. Given $r \geq 0$, suppose $\mu < \nu$ are the extreme indices $n$ for which $v_p(a_n)+nr = \underset{j \geq 0}\min \, v_p(a_j)+jr$. Then $f$ has exactly $\nu-\mu$ zeros (counting multiplicities) on the sphere $\{x \in \C_p : v_p(x) = r\}$.
\end{theorem}
Thm. \ref{thm:newton_polygon} allows us to determine how many $p$-adic zeros we are searching for and provides a certificate when we have found them all.

Finally, we state the $p$-adic Schanuel Conjecture (see~\cite{Calegari_Mazur_2009,Mariaule}). 
\begin{conjecture}[The $p$-adic Schanuel Conjecture]
Let $n \geq 1$ and $t_1, \dots , t_n \in \C_p$ linearly independent over $\Q$, such that $v_p(t_i) \geq \frac{1}{p-1}$ for all $1 \leq i \leq n$). Then 
\begin{align*}
    \mathrm{trdeg}_\Q \ \Q(t_1 , \dots , t_n , \exp(t_1), \dots , \exp(t_n)) \geq n
\end{align*}
where $\mathrm{trdeg}_\Q \ \K$ denotes the transcendence degree of $\K$ over $\Q$.
\end{conjecture}
\section{Decidability of the \texorpdfstring{$p$}{p}-adic Skolem Problem}
\subsection{Informal Outline of the Algorithm}
In this section we prove Thm.~\ref{thm:pDecidable}, assuming the $p$-adic Schanuel Conjecture. 
We start with an informal description of an algorithm that attempts to find the 
$p$-adic zeros of an LRS $\LRS{u}$ using only brute-force search and Hensel's Lemma.  We note the problem with this approach, which motivates 
the subsequent development involving the $p$-adic Schanuel Conjecture.

Let $f_\ell :\Z_{p}\rightarrow \Z_p$ be an interpolant of a given LRS $\LRS{u}$.  
We would like to compute specifications of all the zeros of $f_\ell$ in $\Z_p$. 
The idea is to search for zeros lying in the residue classes of $\Z_p$ modulo $p^r$ for $r=1,2,3,\ldots$.
Consider a representative $z \in \Z_p$ of a residue class of $\faktor{\Z_p}{p^r \Z_p}$. Since $\faktor{\Z_p}{p^r \Z_p} \cong \faktor{\Z}{p^r \Z}$ (e.g. \cite[Ch. II, Prop. 2.4]{neukirch_algebraic_1999}), $z$ may be taken to be an integer in $\{0,1,\dots,p^r-1\}$.

If $f_\ell(z) \not\equiv 0 \bmod p^r$ then the residue class of $z$ does not contain a zero of $f_\ell$, and we can proceed to search other residue classes.  If $f_\ell(z) \equiv 0 \bmod p^r$ and $v_p(f_\ell(z))>2v_p(f_\ell'(z))$ then the residue class contains a unique zero of $f_\ell$ by Hensel's Lemma (Thm. \ref{thm:hensel}). If neither of the above cases hold then the residue class may contain any number of zeros of $f_{\ell}$. Note that we can use Thm.~\ref{thm:newton_polygon} to determine the exact number of extended $p$-adic zeros (lying in the extension $\mathbb C_p$ of $\Z_p$) of $f_{\ell}$ in the residue class of $z$.\footnote{Let $f_\ell(x)=\sum_{n=0}^\infty a_n(x-z)^n$ be a power-series expansion of $f_\ell$ around~$z$ and calculate the respective smallest and largest indices $\mu<\nu$ for which $v_{p}(a_n)+nr= \underset{j \geq 0}\min\, v_{p}(a_j)+jr$ by computing each derivative $f_\ell^{(j)}(z) \bmod p^{k}$ for increasing powers $k$.  By Thm.~\ref{thm:newton_polygon} there are $\nu-\mu$ extended $p$-adic zeros in the same residue class as $z$ modulo $p^r$.} If this number is zero then we can again proceed to consider other residue classes.  If the number is positive then we can refine our search by looking at residue classes modulo $p^{r+1}$ contained in the current class. If all the zeros of $f_{\ell}$ in $\Z_p$ are simple then this search will eventually terminate.  However, if there is a zero in $\Z_p$ of multiplicity two or more then the search will run forever (as the inequality $v_\p(f_\ell(z))>2v_\p(f_\ell'(z))$ will never hold).  The key challenge is thus to identify multiple zeros of $f_{\ell}$ in $\Z_p$ and determine their multiplicity.  




It so happens that for irreducible factors of $f_\ell$, assuming $p$-adic Schanuel's Conjecture, the only possible zeros of multiplicity two or more are rational. This is the content of the main technical results in this section and it allows us to amend the above algorithm so that it always terminates.  Specifically, in parallel with the above-described process, we search by enumeration for rational zeros of $f_\ell$.
We thus find all $p$-adic zeros either by specifying them with Hensel's Lemma or by enumerating rationals and checking directly. We use Thm.~\ref{thm:newton_polygon} to certify that all zeros have thereby been found.
\subsection{Simultaneous Zeros of Coprime Exponential Polynomials}
In this section we prove our main technical result. 
We prove that every common $p$-adic zero of two coprime exponential polynomials is rational, under appropriate conditions, assuming the $p$-adic Schanuel Conjecture.  
It is interesting to compare this result with \cite[Proposition 3.4]{Chonev_zeros_2023}, which shows that two coprime exponential polynomials 
over $\mathbb C$ in which all numerical constants are algebraic have no common zeros whatsoever, assuming Schanuel's Conjecture over $\mathbb C$.

Let $\K$ be a number field with ring of integers $\O$, and $\p \subseteq \O$ be a prime ideal lying above prime $p \in \N$. Choose $\lambda_1,\dots,\lambda_s \in \O$ multiplicatively independent such that the $\p$-adic logarithms $\log \lambda_1, \dots, \log \lambda_s$ are defined. Note then that $\log \lambda_1, \dots, \log \lambda_s$ are linearly independent over $\Q$. Let $P \in \K[x,y_1,\dots,y_s,z_1,\dots,z_s] = \K[x,\bm y, \bm z]$ be a polynomial. Define its associated exponential polynomial by
\begin{align*}
    f_P(x) = P(x,\log \lambda_1, \dots, \log \lambda_s, \exp(x\log \lambda_1) , \dots, \exp(x\log \lambda_s)) \, .
\end{align*}

We first prove a preparatory lemma.
\begin{lemma} \label{lem:prep-lem}
Let $t\in \O_{\C_p}$, and let 
\begin{align*}
    T = \{t , \log \lambda_1, \dots, \log\lambda_s, \exp(t \log \lambda_1), \dots, \exp(t \log \lambda_s)\} \, .
\end{align*}
Then assuming the $p$-adic Schanuel Conjecture, 
\begin{enumerate}[(i)]
    \item If $\mathrm{trdeg}_\Q \Q(T) \leq 2s-1$, then there exist $a_i,b_i \in \Z$ not all zero such that $\sum_{i=1}^s (a_i + b_i t)\log \lambda_i = 0$,
    \item If $\mathrm{trdeg}_\Q \Q(T) \leq 2s-2$, then $t \in \Q$.
\end{enumerate}
\end{lemma}

\begin{proof}
Let $S:=\{\log \lambda_1,\ldots,\log \lambda_s,t\log \lambda_1,\ldots,t\log \lambda_s\}$.  Then
     $\mathrm{trdeg}_{\mathbb Q}(T) =
\mathrm{trdeg}_{\mathbb Q}(S\cup \exp(S))$.
For Part~(i),
suppose that $\mathrm{trdeg}_{\mathbb Q}(T) \leq 2s-1$. Then
$\mathrm{trdeg}_{\mathbb Q}(S\cup \exp(S)) \leq 2s-1$ and hence,
by the $p$-adic Schanuel Conjecture, there is a non-trivial rational linear relation
$\sum_{i=1}^s (a_i+b_it)\log\lambda_i=0$, and by clearing denominators we can take $a_i,b_i \in \Z$.

For Part~(ii),
suppose that $\mathrm{trdeg}_{\mathbb Q}(T) \leq 2s-2$. Observe that if $b_i = 0$ for all $i$, then $\sum_{i=1}^s a_i \log \lambda_i = 0$ contradicts the $\Q$-linear independence of $\log \lambda_1, \dots, \log \lambda_s$. Therefore, we may assume without loss of generality that $b_1 \neq 0$. If $S' = S \setminus\{t \log \lambda_1\}$ were $\Q$-linearly independent, then the $p$-adic Schanuel Conjecture would imply $\mathrm{trdeg}_\Q \Q(T) \geq \mathrm{trdeg}_\Q \Q(S' \cup \exp(S')) \geq 2s-1$, which contradicts the assumption that $\mathrm{trdeg}_\Q \Q(T) \leq 2s-2$. Thus, $S'$ is not $\Q$-linearly indepndent, meaning there are two independent $\mathbb Q$-linear relations on $S$:
$\sum_{i=1}^s (a_i+b_tt) \log \lambda_i=0$ and $\sum_{i=1}^s (c_i+d_it)\log \lambda_i=0$, where $b_1 \neq 0$ and $d_1 = 0$.

Write
\[ A:=\sum_{i=1}^s a_i \log \lambda_i \quad
  B:=\sum_{i=1}^s b_i \log \lambda_i \quad
  C:=\sum_{i=1}^s c_i \log \lambda_i \quad
  D:=\sum_{i=1}^s d_i \log \lambda_i \, . \]
Then $A+tB=C+tD=0$ and so
$t(AD-BC)=0$.  We may suppose that $t \neq 0$, and thus that 
$AD-BC=0$.  By the $p$-adic Schanuel conjecture, $\log \lambda_1,\ldots,\log \lambda_s$ are algebraically independent, so
the polynomial
\[ \left(\sum_{i=1}^s a_iX_i \right)\left(\sum_{j=1}^s d_jX_j \right)- \left(\sum_{i=1}^s b_iX_i \right) \left(\sum_{j=1}^s c_jX_j\right)  
  \]
  is identically zero.

By unique factorisation in $\mathbb Q[X_1,\ldots,X_s]$,
we either have that $(d_1,\ldots,d_s)=q(b_1,\ldots,b_s)$
or $(d_1,\ldots,d_s)=q(c_1,\ldots,c_s)$ for some rational number $q$.

The former case is impossible as $b_1 \neq 0$ and $d_1 = 0$, while in the latter case we have $C+tD = \sum_{i=1}^s c_i (1+qt)\log \lambda_i = 0$. If $1+qt = 0$ then $t \in \Q$ as required; otherwise $\sum_{i=1}^s c_i \log \lambda_i = 0$, which contradicts $\Q$-linear independence of $\log \lambda_1, \dots, \log \lambda_s$.

\end{proof}

Our main result is the following.\footnote{This corrects the statement published in the conference proceedings \cite[Lemma 10]{bacik_padic_2026}. That statement is incorrect as evidenced by the counterexample $P = z_1 - \lambda_2$ and $Q = z_1 - \lambda_2 + y_2 - xy_1$, with $f_P\left(\frac{\log \lambda_2}{\log \lambda_1} \right) = f_Q \left(\frac{\log \lambda_2}{\log \lambda_1} \right) = 0$.}

\begin{restatable}{theorem}{mainLemma}\label{lem:mainLemma}
Let $P,Q \in \K[x,\bm y, \bm z]$ be coprime polynomials with associated exponential polynomials $f_P,f_Q$, and assume each $z_i$ appears with positive degree in either $P$ or $Q$. Then assuming the $p$-adic Schanuel Conjecture, for any $t \in \O_{\C_p}$,
\begin{enumerate}[(i)]
    \item if $f_P(t) = f_Q(t) = 0$, then there exist $a_i,b_i \in \Z$ not all zero such that $\sum_{i=1}^s (a_i + b_i t)\log \lambda_i = 0$,
    \item if $P,Q \in \K[x,\bm z]$ then $t \in \Q$,
    \item if $P \in \K[x,\bm z]$ is irreducible and $f_Q = f_P'$, then $t \in \Q$.
\end{enumerate}
\end{restatable}

\begin{proof}
\textbf{Part (i):} Suppose $t \in \O_{\C_{p}}$ satisfies $f_P(t) = f_Q(t) = 0$. Assume throughout the proof that $t \not \in \Q$, otherwise all parts of the theorem trivially hold, and define
\begin{align*}
    \bm \tau = (t,\log \lambda_1, \dots, \log \lambda_s, \exp(t \log \lambda_1), \dots , \exp(t \log \lambda_s)) \, .
\end{align*}


Now $f_P(t) = f_Q(t) = 0$ implies that $T = \{t, \log \lambda_1, \dots , \log \lambda_s, \exp(t \log \lambda_1), \dots , \exp(t \log \lambda_s)\}$ is comprised of at most $2s-1$ algebraically independent elements. Indeed, pick some element $\sigma$ of $\{x,y_1,\dots,y_s,z_1,\dots,z_s\}$ with positive degree in $P$. Then $f_P(t)=0$ implies that the component of $T$ corresponding to $\sigma$ is algebraic over the remaining components of $T$. Now $f_Q(t) = 0$ implies that the remaining components of $T$ are algebraically dependent. Indeed, if $\sigma$ does not appear in $Q$ then this is obviously true, otherwise since $P$ and $Q$ are coprime, the multivariate resultant $\text{Res}_\sigma(P,Q)$ is a non-zero polynomial in the remaining components of $\{x,y_1, \dots,y_s,z_1,\dots,z_s \}\setminus\{\sigma\}$ with a zero at $ \bm \tau$ (see, e.g., \cite[pp.~163--164]{cox_ideals_2015}). Thus
\begin{align} \label{eqn:trdeg2}
    \mathrm{trdeg}_\Q \Q(t,\log \lambda_1 , \dots , \log \lambda_s, \exp(t\log\lambda_1) , \dots , \exp(t\log\lambda_s)) \leq 2s-1
\end{align}
so by \cref{lem:prep-lem} (i), we have $a_i,b_i \in \mathbb Z$ not all zero such that $\sum_{i=1}^s (a_i + b_i t) \log \lambda_i = 0$, which proves item 1 of the theorem.

Now, for parts (ii),(iii) of the theorem, by \cref{lem:prep-lem} (ii) it suffices to prove that $\mathrm{trdeg}_\Q \Q(T) \leq 2s-2$. Let $R := \sum_{i=1}^s (a_i + b_i x)y_i$. Observe that if $b_i =0$ for all $i$, then $\sum_{i=1}^s a_i \log \lambda_i = 0$ contradicts $\Q$-linear independence of $\log \lambda_1, \dots , \log \lambda_s$, so without loss of generality $b_1 \neq 0$.

\textbf{Part (ii):} in the same way as the proof of part 1, $f_P(t) = f_Q(t) = 0$ implies that 
\begin{align*}
    \mathrm{trdeg}_\Q \Q(t,\exp(t\log \lambda_1),\dots, \exp(t\log \lambda_s)) \leq s-1
\end{align*}
and $f_R(t) = 0$ implies that some $\log \lambda_i$ is algebraic over $\{t,\log \lambda_1,\dots,\log \lambda_s\} \setminus \{\log \lambda_i\}$. These two facts together entail that $\mathrm{trdeg}_\Q \Q(T) \leq 2s-2$, so we are done by \cref{lem:prep-lem} (ii).

\textbf{Part (iii):} in this scenario we have
\begin{align*}
Q = \partial_xP + \sum_{i=1}^s y_i z_i \partial_{z_i}P \, .
\end{align*}
We may assume that each $z_1,\dots,z_s$ appears in $P$ with positive degree, as otherwise we simply pass to the appropriate subring of $\K[x,\bm y, \bm z]$ containing only the $z_i$ appearing in $P$.\footnote{Note that if no $z_i$ appears in $P$ then $P = f_P$ is irreducible, so $f_Q(t) = f_P'(t) = 0$ is impossible.} Thus, every $y_i$ appears with positive degree in $Q$. Also, since $b_1 \neq 0$, we also have that $y_1$ appears with positive degree in $R$. Then we may compute
\begin{align*}
\mathrm{Res}_{y_1}(Q,R) = \left(\sum_{i=2}^s (a_i+b_ix)y_i \right)z_1 \partial_{z_1}P - (a_1 + b_1 x) \left( \partial_x P + \sum_{i=2}^s y_i z_i \partial_{z_i} P \right) \, .
\end{align*}
Suppose, for a contradiction, that $P$ and $\mathrm{Res}_{y_1}(Q,R)$ are not coprime. Then $P \mid \mathrm{Res}_{y_1}(Q,R)$, and so, setting all $y_1 = \dots = y_s = 0$, we have 
\begin{align*}
P \mid (a_1+b_1 x) \partial_x P \, .
\end{align*}
However, since $z_1$ appears in $P$ with positive degree, $P \nmid (a_i +b_i x)$, so\footnote{Recall that in UFDs such as polynomial rings over a field, irreducible elements are prime.} $P \mid \partial_x P$; but since $\deg_x \partial_x P < \deg_x P$, this implies $\partial_x P = 0$.

Now, for each $j \in \{2, \dots ,s \}$, setting $y_i = 0$ for all $i\neq j$ implies
\begin{align*}
    P \mid y_j \left( (a_j+b_j x) z_1 \partial_{z_1} P - (a_1 + b_1 x) z_j \partial_{z_j} P \right)  \, .
\end{align*}
Since $P \nmid x$ and $P \nmid y_j$, we have
\begin{align*}
P \mid (a_j z_1 \partial_{z_1} P - a_1 z_j \partial_{z_j} P)  \quad , \quad P \mid (b_j z_1 \partial_{z_1} P - b_1 z_j \partial_{z_j} P)
\end{align*}
which implies
\begin{align*}
P \mid (a_j b_1 - a_1 b_j)z_j \partial_{z_j}P \, .
\end{align*}
Now if $a_jb_1 - a_1 b_j = 0$ for all $j$, we would have $a_j/b_j = a_1/b_1$ for all $j$, and so $(a_1 + b_1 x)$ would divide $R$. Then $f_R(t) = 0$ implies that $a_1 + b_1 t = 0$, so $t \in \Q$. 

Otherwise, there exists $j$ such that $a_j b_1 - a_1 b_j \neq 0$, so $P \mid z_j \partial_{z_j}P$. If $P \mid z_j$ then $P = z_j$ contradicting $f_P(t) = 0$. So we must have $P \mid \partial_{z_j} P$, which is impossible as $\deg_{z_j} \partial_{z_j} P < P$. 

Thus, we have proven that $P$ and $\mathrm{Res}_{y_1}(Q,R)$ are coprime. We are finally ready to conclude the proof; $f_R(t) = 0$ implies that some $\log \lambda_i$ is algebraic over $\Q(T \setminus \{\log \lambda_i\})$. Next, $\mathrm{Res}_{y_1}(Q,R)$ is a non-zero polynomial which is zero at the point $\bm \tau$, showing that some $\alpha \in T \setminus \{\log \lambda_i\}$ is algebraic over $\Q(T\setminus \{\log \lambda_i,\alpha\})$. Finally, let $\sigma \in \{x,y_1,\dots,y_s,z_1,\dots,z_s\}$ be the variable corresponding to $\alpha$. If $\sigma$ does not appear in $P$ with positive degree, then $f_P(t) = 0$ implies some element of $T\setminus \{\log \lambda_i, \alpha\}$ is algebraic over the remaining elements, which proves that $\mathrm{trdeg}_\Q \Q(T) \leq 2s-2$ as required. 

Otherwise, if $\sigma$ does appear in $P$ with positive degree, then instead the vanishing of the non-zero polynomial $\mathrm{Res}_\sigma(P,\mathrm{Res}_{y_1}(Q,R))$ at $\bm \tau$ proves that $\mathrm{trdeg}_\Q \Q(T) \leq 2s-2$, since $\mathrm{Res}_\sigma(P,\mathrm{Res}_{y_1}(Q,R))$ does not depend on $\sigma$ or $y_1$. 
\end{proof}
There are several interesting consequences.
\begin{remark}
Call $f_P$ an exponential polynomial with algebraic coefficients if $y_i$ does not appear in $P$ with positive degree for $1 \leq i \leq s$. Then \cref{lem:mainLemma} shows that any irrational zero $t$ of an exponential polynomial with algebraic coefficients has a unique irreducible exponential polynomial with algebraic coefficients that it is a root of (if $t$ is algebraic this is just its normal minimal polynomial). Therefore, just as in the case of algebraic numbers, we may associate to any irrational \emph{exponential-algebraic number}  (ie., a root of an exponential polynomial with algebraic coefficients) a unique \emph{minimal} exponential polynomial with algebraic coefficients that it is a root of.
\end{remark}
\begin{remark} \label{rem:complex_remark}
  The result also holds if we work instead with Schanuel's Conjecture on elements of $\C$, and logarithms and exponentials over $\C$.
\end{remark}
\begin{remark} \label{rem:algebraicZeros}
By taking $Q$ to be an integer polynomial in $x$ in the statement of \cref{lem:mainLemma}, one sees that irrational algebraic numbers cannot be zeros of irreducible exponential polynomials that are not polynomials in the usual sense.
\end{remark}

The application of this result to LRS is the following corollary.
\begin{restatable}{corollary}{multCor} \label{cor:multCor}
Let $\K$, $\lambda_1, \dots , \lambda_s$ be as in \cref{lem:mainLemma}. Let $P \in \K[x, \bm y, \bm z]$ be non-zero and irreducible and such that $y_1, \dots y_s$ do not appear with positive degree in $P$. Assuming the $p$-adic Schanuel Conjecture, 
if $t \in \O_{C_{p}}$ is a zero of $f_P$ with multiplicity $\geq 2$ then $t \in \Q$. 
\end{restatable}
\begin{proof}
Let $Q \in \K[x,\bm y, \bm z]$ be such that $f_Q = f_P'$, then as soon as we prove $P,Q$ are coprime, the result follows from item 3 of \cref{lem:mainLemma}, as $t$ being a zero of $f_P$ with multiplicity 2 implies $f_P(t) = f_P'(t) = 0$. 

Since $P$ is irreducible, $P,Q$ only fail to be coprime if $P \mid Q$. Note that
\begin{align*}
    Q = \partial_xP + \sum_{i=1}^s y_i z_{i} \partial_{z_i}P
\end{align*}
Since $P$ has no $y_i$ component for each $1 \leq i \leq s$, we have $P|_{y_i = 0, \, i \in I}=P$ is a non-zero polynomial dividing $Q|_{y_i = 0, \, i \in I}$ for any subset $I \subseteq \{1, \dots , s\}$. In particular this implies that $P \mid \partial_xP$, hence $\partial_xP = 0$ since it has strictly lower total degree than $P$. Now using $I = \{1, \dots , s\} \setminus i$ for each $i$, we also have $P \mid y_i z_{i}\partial_{z_i}P$. If $P \mid z_i$ then $f_P$ has no zeros, contradicting our starting assumption. Therefore we must have $P \mid \partial_{z_i}P$, implying $\partial_{z_i}P = 0$ since it has strictly smaller total degree than $P$. We conclude that $P$ must be constant as all of its partial derivatives vanish. This is impossible as we assume $f_P$ has a zero, so $P$ and $Q$ are coprime.
\end{proof}
We may now ask, does there exist a rational zero of an exponential polynomial $f$ with irreducible underlying multivariate polynomial, with multiplicity greater than 1? The following example shows the answer is yes.  
\begin{example}
Let $Q(x,y) = x^2(x-1) + y^2$. Then $(x-1)$ is a prime ideal in $\overline \Q[x]$ which divides $x^2(x-1)$ exactly once and does not divide the coefficient of $y^2$, so by Eisenstein's criterion $Q(x,y)$ is irreducible in $\overline \Q[x][y] = \overline \Q[x,y]$. Now let $f(x) = Q(\exp(x \log 2)-1,\exp(x \log 3)-1)$. It is easy to see that $0$ is a multiplicity 2 root of $f$.
\end{example}
\subsection{The Algorithm and the Proof of Termination and Correctness}
With Corollary \ref{cor:multCor} in hand, we may amend the algorithm outlined earlier for finding the $p$-adic zeros of an LRS. First we define a recursive subroutine given by Algorithm \ref{alg:zeroSearch}, which finds specifications of all the zeros in $\O_\p$ of a suitable analytic function $f: \O_\p \to \O_\p$.

Next, the idea is as follows. Following the process explained in Section \ref{sec:interpolation}, we find the $p$-adic interpolants of a given LRS $\LRS{u}$, that is, analytic functions interpolating each subsequence $u_{Nn+\ell}$ for each $0 \leq \ell \leq N-1$, for some integer $N$. Since any interpolant $f$ maps $\Z$ to $\Z$, we have by continuity that $f$ maps $\Z_p$ to $\Z_p$. We compute a multiplicatively independent set of algebraic numbers $\{\alpha_1, \dots, \alpha_s\}$ that generates the characteristic roots of $\LRS{u}$ as explained in Section \ref{sec:interpolation}, and for each interpolant $f$ compute the associated Laurent polynomial $P$ such that $f(n) = P(n,\alpha_1^{Nn},\dots,\alpha_s^{Nn})$ for all $n \in \Z$.

Now, factorise $P = RP_1^{r_1}\dots P_m^{r_m}$, where $P_1,\dots,P_m$ are irreducible polynomials, and $R$ is some Laurent monomial. Next, pass the associated analytic function $f_{P_i} : \O_\p \to \O_\p$ to Algorithm \ref{alg:zeroSearch} to find specifications of all extended $p$-adic zeros in $\O_\p$ of each $f_{P_i}$. Unfortunately, we cannot guarantee that $f_{P_i}$ maps $\Z_p$ to $\Z_p$ so we require one more step to identify the $p$-adic zeros (i.e. those in $\Z_p$). Note however, that for an irrational zero $x \in \O_\p$ of $f_{P_i}$ that has multiplicity 1, $x$ will also be a zero of $f^{(r_i-1)} : \Z_p \to \Z_p$ of multiplicity 1. Therefore for a suitable approximation $y$ of $x$, we can identify $x$ as a zero of $f^{(r_i-1)}$ using Hensel's lemma when we certify $v_p(f^{(r_i-1)}(y)) > 2v_p(f^{(r_i)}(y))$. Since $f^{(r_i-1)}$ does map $\Z_p$ to $\Z_p$, if $y \in \Z_p$ then Hensel's lemma certifies that $x \in \Z_p$. Otherwise, if $x \in \O_p \setminus \Z_p$, one can certify this by computing its $\p$-adic series expansion. Thus, to find all the zeros of $f$ in $\Z_p$, simply run through each zero $x \in \O_p$ of each $f_{P_i}$, and in parallel do the following:
\begin{enumerate}
    \item Compute approximations $y$ of $x$ and check if one may take $y \in \Z$ and whether the Hensel condition $v_p(f^{(r_i-1)}(y))> 2v_p(f^{(r_i)}(y))$ holds.
    \item Enumerate rationals $q \in \Q$ matching the $p$-adic expansion of $x$ and check algebraically whether $f_{P_i}(q) = 0$.
\end{enumerate}
Checking for rationals deals with the situation when a rational zero is shared by several $f_{P_i}$. Pseudocode for the entire algorithm is given in Algorithm \ref{alg:fullAlg}.

\begin{algorithm} 
\SetKwFunction{zeroSearch}{zeroSearch}
\SetKw{And}{and}
\SetKwData{zerosFound}{zerosFound}
\SetKwData{multCount}{multCount}
\SetKwData{henselZeros}{henselZeros}
\SetKwData{ratZeros}{ratZeros}
\SetKwData{residues}{residues}
\SetKw{Output}{output:}
\SetKwFor{Parallel}{in parallel until}{do}{}
\SetKwFor{Until}{until}{do}{}
\caption{The \protect\zeroSearch subroutine} \label{alg:zeroSearch}
\KwIn{A non-zero $\p$-adic analytic function $f: \O_{\p} \to \O_{\p}$, $\pi$ a computed algebraic number that is a uniformiser for $\p$, and $A$ a set of representatives of the residue field ${\O_{\p}}/{\p}$ computed as described in Section \ref{sec:interpolation}.}
\KwOut{A set \henselZeros of tuples $(x,\nu)$ such that there is a unique simple zero $z$ of $f$ with $x \equiv z \mod \pi^\nu$, and a set \ratZeros of tuples $(a/b,n)$ where $a/b$ is a zero of $f$ of multiplicity $n$}
\BlankLine
\Def{\zeroSearch$(f)$}{
\henselZeros $\leftarrow$ empty set\;
\ratZeros $\leftarrow$ empty set\;
\residues $\leftarrow \{(0,0)\}$\;
\While{$\residues \neq \emptyset$}{
\ForEach{$(z,r) \in \residues$}{
\ForEach{$a \in A$}{
$z_{r+1} \leftarrow z + a \pi^r$\;
\If{$v_\p(f(z_{r+1})) > 2v_\p(f'(z_{r+1}))$ \And $\nexists (x,\nu) \in \henselZeros$ \textup{with} $x \equiv z_{r+1} \mod \pi^\nu$}{
append $(z_{r+1},v_\p(f'(z_{r+1}))+1)$ to \henselZeros\;
}
\Else{
\ForEach{$a/b \in \Q$ \textup{with} $a,b$ \textup{coprime}, $|a|,|b| \leq r$
}{
$n \leftarrow$ multiplicity of $a/b$ as a zero of $f$\;
\If{$f(a/b) = 0$ \And $(a/b,n) \not\in \ratZeros$ \And $\nexists (x,\nu) \in \henselZeros$ \textup{with} $a/b \equiv x \bmod \pi^\nu$}{
Append $(a/b,n)$ to \ratZeros\;
}
}
}
$n \leftarrow$ the number (counting multiplicity) of extended $p$-adic zeros $\equiv z_{r+1} \bmod \pi^{r+1}$ (computed using Thm.~\ref{thm:newton_polygon})\;
$m \leftarrow$ the sum of multiplicities of all zeros $x$ in $\henselZeros$ or $\ratZeros$ with $x \equiv z_{r+1} \bmod \pi^{r+1}$\;
\lIf{$n=m$}{\Next $a$}
\lElse{Append $(z_{r+1},r+1)$ to \residues}
}
Remove $(z,r)$ from \residues\;
}
}
\Output{\henselZeros,\ratZeros}
}

\end{algorithm}
\begin{algorithm} 
\SetKwFunction{zeroSearch}{zeroSearch}
\SetKwData{candidates}{candidates}
\SetKwData{Fzeros}{Fzeros}
\SetKwData{henselZeros}{henselZeros}
\SetKwData{ratZeros}{ratZeros}
\SetKwData{uhenselZeros}{u\_henselZeros}
\SetKwData{uratZeros}{u\_ratZeros}
\SetKw{Output}{output:}
\SetKw{And}{and}
\caption{An algorithm to compute the $p$-adic zeros of an LRS $\LRS{u}$}
\KwIn{A non-degenerate LRS $\LRS{u}$, number field $\K$ containing $\LRS{u}$ and its characteristic roots, and prime ideal $\p \subseteq \O$ lying above prime $p$ with $v_\p(\lambda) = 0$ for all characteristic roots $\lambda$. We let $\pi$ be a computed uniformiser for $\p$, and $A$ a set of representatives of the residue field ${\O_{\p}}/{\p}$} \label{alg:fullAlg}
\KwOut{All the $p$-adic zeros of $\LRS{u}$ in the form of two lists. First, a list \uhenselZeros of tuples $(x,\nu,f_i,\ell)$ where $x, \nu, \ell \in \Z$, $f_i : \O_\p \to \O_\p$ is a factor of the $\ell$th interpolant $F_\ell$ of $\LRS{u}$ and there is a unique simple $p$-adic zero $z \in \Z_p$ of $f_i$ (and so a zero of $F_\ell$) such that $x \equiv z \mod p^\nu$. Secondly, a list \uratZeros of tuples $(a/b,\ell)$ for $a/b \in \Q$ with $F_\ell(a/b) = 0$.}
Compute a multiplicatively independent set $\{\alpha_1, \dots , \alpha_s\}$ of algebraic numbers that generates every characteristic root as a monomial (replace $\K$ by $\K(\alpha_1,\dots,\alpha_s)$ if necessary)\;
Compute least integer $N > 0$ such that $v_\p(\lambda_i^{N}-1) > \frac{e}{p-1}$ for all $1\leq i \leq s$, where $e = e_{\p/p}$ is the ramification index\;
\uhenselZeros $\leftarrow$ empty list\;
\uratZeros $\leftarrow$ empty list\;
\ForEach{$0 \leq \ell \leq N-1$}{
Compute Laurent polynomial $Q$ such that $u_{Nn+\ell} = Q(n,\alpha_1^{Nn},\dots,\alpha_s^{Nn})$\;
$F(x) \leftarrow Q(x,\exp(x\log \alpha_1^{N}),\dots,\exp(x\log \alpha_s^{N}))$\;
Factorise $Q = RP_1^{r_1}\dots P_m^{r_m}$ for irreducible polynomials $P_1, \dots, P_m$ and Laurent monomial $R(\alpha_1^n,\dots,\alpha_s^n)$\;
\ForEach{\textup{$P_i$}}{
$f_{P_i}(x) \leftarrow P_i(x,\exp(x \log \alpha_1^{N}),\dots,\exp(x \log \alpha_s^{N}))$\;
\henselZeros,\ratZeros $\leftarrow$ \zeroSearch$(f_{P_i})$\;
\ForEach{$(a/b,n) \in \ratZeros$}{
Append $(a/b,\ell)$ to \uratZeros\;
}
\ForEach{$(x,\nu) \in \henselZeros$}{
\ForEach{$n \geq 1$}{
$x_n \leftarrow$ approximation of the zero $z$ of $f_{P_i}$ corresponding to $(x,\nu)$ such that $x_n \equiv z \mod \pi^n$\;
\If{$\nexists y \in \{0,1,\dots,p^n-1\}$ \textup{with} $y \equiv x_n \mod \pi^n$}{
\Next $(x,\nu)$\;
}
\Else{
\ForEach{$a/b \in \Q$ \textup{with} $|a|,|b| \leq n$}{
\If{$F(a/b) = 0$ \And $a/b \equiv x \mod \pi^\nu$}{
Append $(a/b,\ell)$ to \uratZeros\;
\Next $(x,\nu)$\;
}
\Else{
\Next $a/b$\;
}
}
$y \leftarrow $ integer in $\{0,1,\dots,p^n-1\}$ such that $y \equiv x_n \mod \pi^n$\;
\If{$v_p(F^{(r_i-1)}(y)) > 2v_p(F^{(r_i)}(y))$}{
Append $(y,\nu,f_{P_i},\ell)$ to \uhenselZeros
\Next $(x,\nu)$\;
}
\Else{
\Next $n$
}
}
}
}
}
}
\Output{\uhenselZeros,\uratZeros}
\end{algorithm}


Theorem~\ref{thm:pDecidable} follows from the next result.
\begin{proposition}
Assuming the $p$-adic Schanuel  Conjecture, Algorithm \ref{alg:fullAlg} always terminates.  Upon termination it outputs all $p$-adic zeros of the  given LRS $\LRS{u}$.
\end{proposition}
\begin{proof}
Since each factor $P_i$ defined on line 8 of Algorithm \ref{alg:fullAlg} is irreducible, by Cor. \ref{cor:multCor} the $p$-adic Schanuel Conjecture implies that every zero in $\O_{\p}$ of $f_{P_i}$ either has multiplicity 1 or is rational. Therefore, $\zeroSearch(f_{P_i})$ will terminate.

Indeed, for each pair $(z,r)$ (line 6, Algorithm \ref{alg:zeroSearch}), there are 3 possibilities for extended $p$-adic zeros $x \in \O_{\C_p}$ detected on line 16 using Thm. \ref{thm:newton_polygon}.
\begin{enumerate}
    \item $x \in \O_{\C_p} \setminus \O_\p$. In this case $x$ will no longer be detected in higher iterations $(z',r')$ for large enough $r' > r$ as $\sup\{v_\p(x-y) : y \in \O_\p\} < \infty$. Therefore this case may be discarded for large enough $r$.
    \item $x \in \Q$ and has multiplicity $>1$. In this case $x = a/b$ for some $a,b \in \Z$ coprime and for large enough $r$ this will be found in the search on line 12. Note that the multiplicity of $a/b$ as a root of $f_{P_i}$ may be found as each $f_{P_i}^{(j)}(a/b)$ is a polynomial in $\log \alpha_1^N, \dots, \log \alpha_s^N$ with algebraic coefficients, so pick the smallest $j$ such that this polynomial is not identically zero and $p$-adic Schanuel's conjecture implies $f_{P_i}^{(j)}(a/b) \neq 0$. One may approximate $f_{P_i}^{(j)}(a/b)$ $p$-adically to certify it is non-zero.
    \item $x \in \O_\p$ and has multiplicity 1. In this case $x$ will be detected by Hensel's lemma using the check on line 9. Note that this check is done in the following way. If $z_{r+1} \in \Q$ then one can check algebraically (using the $p$-adic Schanuel Conjecture) whether $f(z_{r+1})$ or $f'(z_{r+1}) = 0$, and if they are non-zero then one can compute their valuations by computing sufficiently many $\p$-adic digits. Otherwise if $z_{r+1} \not\in \Q$ then by Cor. \ref{cor:multCor} either $f(z_{r+1}) \neq 0$ or $f'(z_{r+1}) \neq 0$ so by computing sufficiently many $\p$-adic digits of both quantities one can determine whether $v_\p(f(z_{r+1})) > 2v_\p(f'(z_{r+1}))$. 
\end{enumerate}
When $\zeroSearch(f_i)$ terminates, lines 14-30 of Algorithm \ref{alg:fullAlg} pick out the zeros identified that actually lie in $\Z_p$ and outputs them. Indeed, for every pair $(x,\nu)$ corresponding to zero $z \in \O_\p$ (line 14, Algorithm \ref{alg:fullAlg}), there are three possibilities:
\begin{enumerate}
    \item $z \in \O_\p \setminus \Z_p$. In this case, since $\sup\{v_\p(z-y) : y \in \Z_\p\}$ is finite, for large enough $n$ there will not exist $y \in \{0,1,\dots,p^n-1\}$ with $y \equiv x_n \mod \pi^n$, so this case is thrown out (lines 17-18).
    \item $z \in \Z_p \setminus \Q$. In this case, item 2 of \cref{lem:mainLemma} ensures $z$ is not a zero of $f_{P_j}$ for $j \neq i$, so $z$ is a zero of $F$ with multiplicity exactly $r_i$. Hence $z$ is confirmed as a zero of $F^{(r_i-1)}$ with Hensel's lemma by the check on line 27, which confirms $z \in \Z_p$.
    \item $z \in \Q$. In this case $z$ could be a zero of several of the $f_{P_1}, \dots, f_{P_m}$, so $z$ would be a zero of $F$ of multiplicity larger than $r_i$, so the if statement on line 27 never triggers. However, for large enough $n$, $z$ would be identified by the checks on lines 20-25.\qedhere
\end{enumerate}

\end{proof}
\begin{remark}
We stress that the  $p$-adic Schanuel Conjecture is required only for termination, not for correctness. When the algorithm terminates, its output is unconditionally correct. 
\end{remark}

\subsection{Simultaneous Skolem Problem}
\cref{lem:mainLemma} has important consequences for the Simultaneous Skolem Problem. Given two non-degenerate LRS $\LRS{u}$, $\LRS{v}$, suppose their exponential-polynomial expansions may be written in terms of multiplicatively independent algebraic numbers $\{\alpha_1, \dots, \alpha_s\}$. Suppose the multivariate Laurent polynomials defined by the respective exponential-polynomial expansions are coprime.\footnote{Note that two LRS that are coprime in the sense of~\cite{Everest2002} are coprime in our sense.} In this case we say that $\LRS{u}$ and $\LRS{v}$ are \emph{coprime}. Equivalently, $\LRS{u}$ and $\LRS{v}$ are coprime if there are no algebraic-valued LRS $\LRS{w}^{(1)},\LRS{w}^{(2)},\LRS{w}^{(3)}$ such that $\LRS{w}^{(1)}$ has order at least 2 and $u_n = w_n^{(1)}w_n^{(2)}, \, v_n = w_n^{(1)}w_n^{(3)}$ for all $n \in \N$. We define the Simultaneous Skolem Problem to be the problem of deciding whether two LRS $\LRS{u}$ and $\LRS{v}$ share an integer zero. 

\SIMP*
\begin{proof}
Pick a prime $p$ not dividing the characteristic roots of $\LRS{u},\LRS{v}$ and consider the LRS $\LRS{w}$ defined by $w_n = v_n^2 + p u_n^2$. For some $N \geq 1$, and all $0 \leq \ell \leq N-1$, their interpolants with respect to $p$ are related by $f_{w,\ell} = f_{v,\ell}^2 + p f_{u,\ell}^2$. For all $x \in \Z_p$, $f_{w,\ell}(x) = 0$ iff $f_{v,\ell}(x) = f_{u,\ell}(x) = 0$. Indeed, if $f_{v,\ell}(x),f_{u,\ell}(x)$ are non-zero then $v_p(f_{v,\ell}(x)^2)$ is even, and $v_p(pf_{u,\ell}(x)^2)$ is odd, so $f_{w,\ell} \neq 0$.


By Thm.~\ref{thm:pDecidable}, we may find all $p$-adic zeros of $\LRS{w}$. By coprimality of $\LRS{u}, \LRS{v}$, all the $p$-adic zeros found must be rational according to item 2 of \cref{lem:mainLemma}. Therefore they may be found by a brute-force guess-and-check search, which in particular identifies all common integer zeros.
\end{proof}
The Simultaneous Skolem Problem arises naturally in the study of the higher-dimensional Kannan-Lipton Orbit Problem \cite{Kannan_orbit_1986,chonev_orbit_2013}. Given a matrix $A \in \Q^{d \times d}$, initial vector $x \in \Q^d$ and subspace $S \subset \Q^d$, the Orbit Problem asks to decide whether there is $n \in \N$ such that $A^n \in S$. Now, one may write $S$ as the orthogonal space of the span of $k$ linearly independent vectors for some $k$, that is $S = \langle v_1,\dots,v_k\rangle^\perp$. Then $A^n x \in S$ if and only if $v_i^T A^n x = 0$ for all $i=1,\dots, k$. Note that $u^{(i)}_n := v_i^T A^n x$ defines $k$ LRS $\LRS{u}^{(1)},\dots, \LRS{u}^{(k)}$ whose recurrence relation is given by the characteristic polynomial of $A$. Therefore the Orbit Problem on $A,x,S$ reduces to the Simultaneous Skolem Problem on $\LRS{u}^{(1)},\dots, \LRS{u}^{(k)}$. Thus Thm. \ref{thm:simultaneous} may be used to decide many cases of the Orbit Problem when $k \geq 2$, however, we do not expand on this further here. 
\subsection{Algebraic-valued LRS}
Throughout the present article we have taken our LRS to be integer valued. The benefit of this is that integer-valued LRS always have interpolants $f : \Z_p \to \Z_p$. However, $p$-adic interpolants of algebraic-valued LRS do not in general map $\Z_p$ to $\Z_p$. 

One can still decide the $p$-adic Skolem Problem for algebraic-valued LRS by reducing to the integer case. First, given an algebraic LRS $\LRS{u}$, assume all its recurrence coefficients and initial values are algebraic integers - if not, then there's a $\lambda \in \overline \Q$ such that $u_n' = \lambda^nu_n$ is algebraic integer valued; consider $u_n'$ instead. Next, define $\K$ to be the number field containing all the characteristic roots and initial values of $\LRS{u}$. Consider 
\begin{align} \label{eqn:norm}
    v_n = N_{\K/\Q}(u_n) = \prod_{\sigma \in \mathrm{Gal(\K/\Q)}}\sigma(u_n)\, .
\end{align}
Then $\LRS{v}$ is an integer LRS, and we may solve the $p$-adic Skolem Problem on $\LRS{v}$. We now claim $\LRS{v}$ has identical $p$-adic zeros to $\LRS{u}$.

Indeed, $x \in \Z_p$ is a $p$-adic zero of an (algebraic) LRS $w_n$ if and only if there's a sequence $n_1,n_2, \dots \in \N$ such that $|u_{n_j}|_p \to 0$ and $|n_j - x|_p \to 0$ as $j \to \infty$. Using \eqref{eqn:norm} and the fact that $|\sigma(y)|_p = |y|_p$ for all $y \in \K$, $\sigma \in \mathrm{Gal}(\K/\Q)$, we see that if such a sequence $\langle n_j \rangle_{j=1}^\infty$ converging to a $p$-adic zero $x$ of $v_n$ exists, then we must have $|u_{n_j}|_p \to 0$ also, so $x$ must be a $p$-adic zero of $\LRS{u}$ also, and vice versa. This proves the claim.
\subsection{The Skolem Conjecture and the Skolem Problem} \label{sec:SkolemConjecture}
As noted earlier, it is open whether there is a Turing reduction between $\mathtt{FSP}(\N)$ and $\mathtt{FSP}(\Z_p)$. 
For a given sequence $\LRS{u}$, if
we happen to choose a prime $p \in \N$ such that all $p$-adic zeros of $\LRS{u}$ are rational, then they can all be identified, and the output of Algorithm \ref{alg:fullAlg} gives a certificate that we have found all the integer zeros. This solves $\mathtt{FSP}(\N)$ for $\LRS{u}$. One may ask whether such a prime always exists. It turns out that a generalisation of this idea is equivalent to the Skolem Conjecture, also known as the Exponential Local-Global Principle.
\begin{conjecture}[The Skolem Conjecture \cite{Sko37}]
Let $\LRS{u}$ be a simple rational LRBS taking values in the ring $\Z[\frac{1}{b}]$ for some integer $b$. Then $\LRS{u}$ has no integer zero iff, for some integer $m \geq 2$ with $\text{gcd}(b,m)=1$, we have that $u_n \not\equiv 0 \bmod m$ for all $n \in \Z$.
\label{conj:skolem}
\end{conjecture}

\begin{theorem}\label{Skolem_conjecture_equivalent}
The Skolem Conjecture is equivalent to the following statement: if $\LRS{u}$ is a simple rational LRBS taking values in the ring $\Z[\frac{1}{b}]$ for some integer $b$, then $\LRS{u}$ has no integer zero iff there exists $N \in \Z_{\geq 1}$ and primes $p_1, \dots , p_t \in \N$ coprime to $b$ such that for all $0 \leq \ell \leq N-1$, there exists $i$ such that $u_{Nn+\ell}$ has no $p_i$-adic zeros.
\end{theorem}
\begin{proof}
Suppose $\LRS{u}$ is a simple rational LRBS of order $d$ taking values in the ring $\Z[\frac{1}{b}]$ for some integer $b$ and that $\LRS{u}$ has no integer zeros (the reverse implication is trivial for both statements). 

Suppose the Skolem Conjecture holds, then there's an integer $m$ with $\text{gcd}(m,b) = 1$ and $u_n \not\equiv 0 \bmod m$ for all $n \in \Z$. Write
\begin{align*}
m = \prod_{i=1}^t p_i^{k_i}.
\end{align*}
Since ${\Z}/{p_i^{k_i}\Z}$ has finitely many elements, the vector $(u_n, u_{n+1}, \dots , u_{n+d-1})$ takes only finitely many values mod $p_i^{k_i}$ and thus eventually repeats. Since $\LRS{u}$ has order $d$, this means that $u_n \bmod p^{k_i}$ is periodic, with some period $N_i$. Let $N = \prod_i N_i$. Then each subsequence $u_{Nn+\ell}$ for $0 \leq \ell \leq N-1$ is constant mod $m$. In particular, for each $0 \leq \ell \leq N-1$ there's $i$ such that $u_{Nn+\ell} \not\equiv 0 \bmod p_i^{k_i}$. Therefore there are no $p_i$-adic zeros of $u_{Nn+\ell}$.

Now suppose there exist $N$, $p_1, \dots , p_t$ as in the theorem statement. Then for each $0 \leq \ell \leq N-1$, the subsequence $u_{Nn+\ell}$ has no $p_i$-adic zeros for some $i$, which means that $v_{p_i}(u_{Nn+\ell}) < k_\ell$ for some integer $k_\ell$. Let $k = \underset{\ell} \max \ k_\ell$. Then the Skolem Conjecture holds for $\LRS{u}$ with $m = \prod_i p_i^k$.
\end{proof}



\section{Concluding Remarks}
\subsection{A Remark on Rational Zeros} \label{sec:rat_zeros}
Note that rational non-integer zeros can be found as $p$-adic zeros for certain $p$, but not always. This is because, informally, interpolating an LRS $\LRS{u}$ using some prime ideal $p$ fixes a definition of each $\lambda_i^{\frac{1}{b}}$ for each $b \in \Z_{\geq 1}$, whereas the notion of rational zero allows $\lambda_i^{\frac{1}{b}}$ to be any $b$-th root of $\lambda_i$. This phenomenon can be seen in the example below.

\begin{example} \label{Tribonacci}
Consider the Tribonacci sequence, defined by $u_{n+3} = u_{n+2} + n_{n+1} + u_n$ and $u_0 = 0$, $u_1 = u_2 = 1$. In \cite{bilu_twisted_2025} it is shown that the set of rational zeros of the Tribonacci sequence is exactly $\{0,-1,-4,-17,1/3,-5/3\}$.

The characteristic polynomial of the Tribonacci sequence splits in $\Z_p$ for $p = 47, 103,199$ (among others). Let the $\ell$-th interpolant with respect to $p$ be denoted $f_{p,\ell}$, that is, the analytic function $f_{p,\ell} : \Z_p \to \Z_p$ such that $f_{p,\ell}(n) = u_{Nn+\ell}$ for $\ell\in\{0,\ldots,N-1\}$. Here we used $N= 46, 51, 198$ for $p=47,103,199$ respectively. We used our tool to compute all the $p$-adic zeros of $\LRS{u}$ for each $p$. Define $\mathcal Z_p$ to be the set of tuples $(z,\ell)$, where $z$ is the $p$-adic zero of $\LRS{u}$ such that $f_{p,\ell}(z) = 0$.
\begin{align*}
\mathcal Z_{47} = \{(0,0),(-2/3, 29), (-1,29), (-2/3,31),(-1,42),(-1,45)\} \\
\mathcal Z_{103} = \{(0,0), (-1/3,13), (-1/3,16),(-1/3,17),(-2/3,17),(-2/3,30), \\ (-2/3,33),(-2/3,34),(-1,34),(-1,47),(-1,50)\}, \\
\mathcal Z_{199} = \{(0,0),(185+195\cdot 199 + 135 \cdot 199^2 + \dots,26), (-1/3,49), (-1/3,62), \\ (-1/3,65), (-1/3,66), (52+63 \cdot 199 + 3 \cdot 199^2 + \dots, 92), (-2/3,115), \\(-2/3,128), (-2/3,131), (-2/3,132), (118+129\cdot 199 + 69 \cdot 199^2 + \dots,158), \\ (-1,181),(-1,194), (-1,197) \} \, .
\end{align*}
For each tuple $(z,\ell)$ corresponding to a rational zero of $\LRS{u}$, recover the rational zero as $Nz+\ell$.
\begin{itemize}
    \item For $p=47$ and $N=46$, all the rational zeros of $\LRS{u}$ are correctly identified and there are no transcendental $47$-adic zeros. 
    \item For $p=103$ and $N=51$, each integer zero gives rise to three $103$-adic zeros. For example, $n = -4$ gives rise to $(-1/3,13), (-2/3,30), (-1,47)$. However, the rational zeros of $\LRS{u}$ do not appear as $103$-adic zeros of $\LRS{u}$. There are also no transcendental $103$-adic zeros of $\LRS{u}$.
    \item For $p=199 $ and $ N=198$, again each integer zero gives rise to three $199$-adic zeros. Similarly to $p=103,$ the rational zeros of $\LRS{u}$ do not appear as $199$-adic zeros, but there are several transcendental $199$-adic zeros of $\LRS{u}$.
\end{itemize}
\end{example}
It is also possible that a rational $p$-adic zero $z$ of a subsequence $u_{Nn+\ell}$ of some LRS $\LRS{u}$ does not correspond to a ``true'' zero of the original sequence, i.e, $Nz + \ell \in \Z$ but $u_{Nz+\ell} \neq 0$.
\begin{example}
Let $u_n = (2+i)^n+(2-i)^n$. Using our tool with $p = 13$, we find that the analytic function $f_6 : \Z_{13} \to \Z_{13}$ satisfying $f_6(n) = u_{12n + 6}$ for all $n \in \N$ has $f_6(-1/2) = 0$. Yet $u_{12\cdot(-\frac{1}{2}) + 6} = u_0 \neq 0$. This is because $\exp(\frac{1}{2} \log((\lambda_i)^{12})) = \pm \lambda_i^6$ for $\lambda_1, \lambda_2$ the characteristic roots of $\LRS{u}$ embedded in $\Z_{13}$. It so happens that in $\Z_{13}$, up to relabelling $\lambda_1,\lambda_2$ we have $\exp(\frac{1}{2} \log((\lambda_1)^{12})) =  \lambda_1^6$, $\exp(\frac{1}{2} \log((\lambda_2)^{12})) = - \lambda_2^6$, so $f_6(-1/2) = 0$. Informally, what went wrong is that roots of unity were introduced by the ``wrong'' square roots of $\lambda_i^{12}$ being chosen. In this situation we say $0$ is a twisted rational zero of $\LRS{u}$. See \cite{bilu_twisted_2025} for a definition and exploration of this phenomenon. It should be noted however, that every genuine integer zero of $\LRS{u}$ will arise as a $p$-adic zero for all primes $p$. Indeed, if $u_m = 0$, then for every integer $N$ there exists $n, \ell$ such that $Nn+\ell = m$, and $n$ will be a zero of any $p$-adic interpolant $f_\ell$ of $u_{Nn+\ell}$ since $\exp(n\log(\lambda^N)) = \exp(\log(\lambda^{Nn})) = \lambda^{Nn}$ for all characteristic roots $\lambda$, so $f_\ell(n) = u_{Nn+\ell} = u_m = 0$.
\end{example}

\section{Implementation and Experimental Analysis} \label{sec:experiments}
\begin{table}[ht]
\resizebox{\linewidth}{!}{%
\begin{tabular}{@{}rrrrrrrrrrrrrrrr@{}}
\toprule
$d$ & Success & Total & \%  & \showclock{0}{45} & 
\multicolumn{7}{l}{Count of instances with $n$ zeros for $n$:}& Avg & \multicolumn{3}{l}{\#Zeros of type} \\
&  (Count) & (Count)  &   & & 0    & 1    & 2    & 3   & 4   & 5  & 6+  & prime & $\mathbb{Z}$ & $\mathbb{Q}\setminus\mathbb{Z}$ & ? \\ \midrule
2     & 8886            & 8886        & 100\%              & 0.6                            & 4322 & 3730 & 625  & 112 & 59  & 20 & 18  & 9.2              & 839      & 108      & 4817     \\
3     & 8921            & 8921        & 100\%              & 2.2                            & 3896 & 2825 & 1526 & 378 & 154 & 21 & 121 & 34.5             & 862      & 3        & 7962     \\
4     & 9080            & 9160        & 99\%               & 11.3                           & 3527 & 3017 & 1687 & 537 & 183 & 28 & 101 & 144.7            & 983      & 2        & 9281     \\
5     & 4892            & 9172        & 53\%               & \textcolor{gray}{25.7}                           & 1830 & 1710 & 903  & 316 & 95  & 13 & 25  & \textcolor{gray}{263.8}            & 568      & 0        & 4803     \\
6     & 934             & 9162        & 10\%               & \textcolor{gray}{30.0}                           & 346  & 327  & 188  & 50  & 15  & 2  & 6   & \textcolor{gray}{252.5}           & 144      & 0        & 1081     \\
7     & 96              & 9201        & 1\%                & \textcolor{gray}{33.3}                           & 29   & 34   & 22   & 8   & 3   &  0  & 0   & \textcolor{gray}{240.9}            & 17       & 0        & 97       \\ \bottomrule
\end{tabular}
}
\caption{Summary of ``Hensel-only'' algorithm, for each order from $2-7$. The \showclock{0}{45} indicates the average running time in seconds for successful cases (failed cases timeout at 60s). Averages are \textcolor{gray}{gray} where skewed by having many timeouts (counts may also be affected, but should be considered relative to the successful cases).\vspace{-0.3cm}
}
\label{table:experiments1}
\end{table}

\begin{table}[ht]
\begin{subtable}[t]{0.6\textwidth}
\setlength\tabcolsep{0pt}
\begin{tabular*}{\linewidth}{@{\extracolsep{\fill}} rrrrrr }
\toprule
$d$ & Avg & StdDev & Max & Min & Timeouts  \\ \midrule
2     & 9                & 5               & 37            & 3             & 0/4435    \\
3     & 35               & 29              & 227           & 5             & 0/4460    \\
4     & 151              & 146             & 1831          & 5             & 0/4578    \\
5     & 817              & 905             & 11677         & 7             & 0/4582    \\
6     & 5962             & 6395            & 45413         & 7             & 1/4581    \\
7     & 18164            & 13268           & \textcolor{gray}{49957}         & 31            & 1842/4602 \\ \bottomrule
\end{tabular*}
\caption{Growth of primes that would be used by the algorithm (mean, standard deviation, max and min of the required prime reported for each LRS order $d$), with timeout of 60s or if the prime would exceed 50000. The max at order 7 (\textcolor{gray}{gray}) can be attributed to the timeout, rather than the true value.}
\label{table:primes}
\end{subtable}
    \hfill
\begin{subtable}[t]{0.35\textwidth}
\begin{tabular}{@{}rrrrr@{}}
\toprule
$d$ & Success & Total & \%    & \showclock{0}{45}     \\ \midrule
2 & 8886    & 8886  & 100\% & 0.9  \\
3 & 8920    & 8921  & 100\% & 3.0  \\
4 & 8613    & 9160  & 94\%  & 19.1 \\
5 & 2       & 9172  & 0\%   & \textcolor{gray}{5.9}  \\
6 & 0       & 9162  & 0\%   &      \\
7 & 0       & 9201  & 0\%   &      \\ \bottomrule
\end{tabular}
\caption{Summary of the full algorithm with 60s timeout.}
\label{table:experiment2}
\end{subtable}
\caption{Values \textcolor{gray}{gray} where skewed by the timeout.}
\end{table}

\begin{figure}[h]
    \centering
    \includegraphics[width=\linewidth]{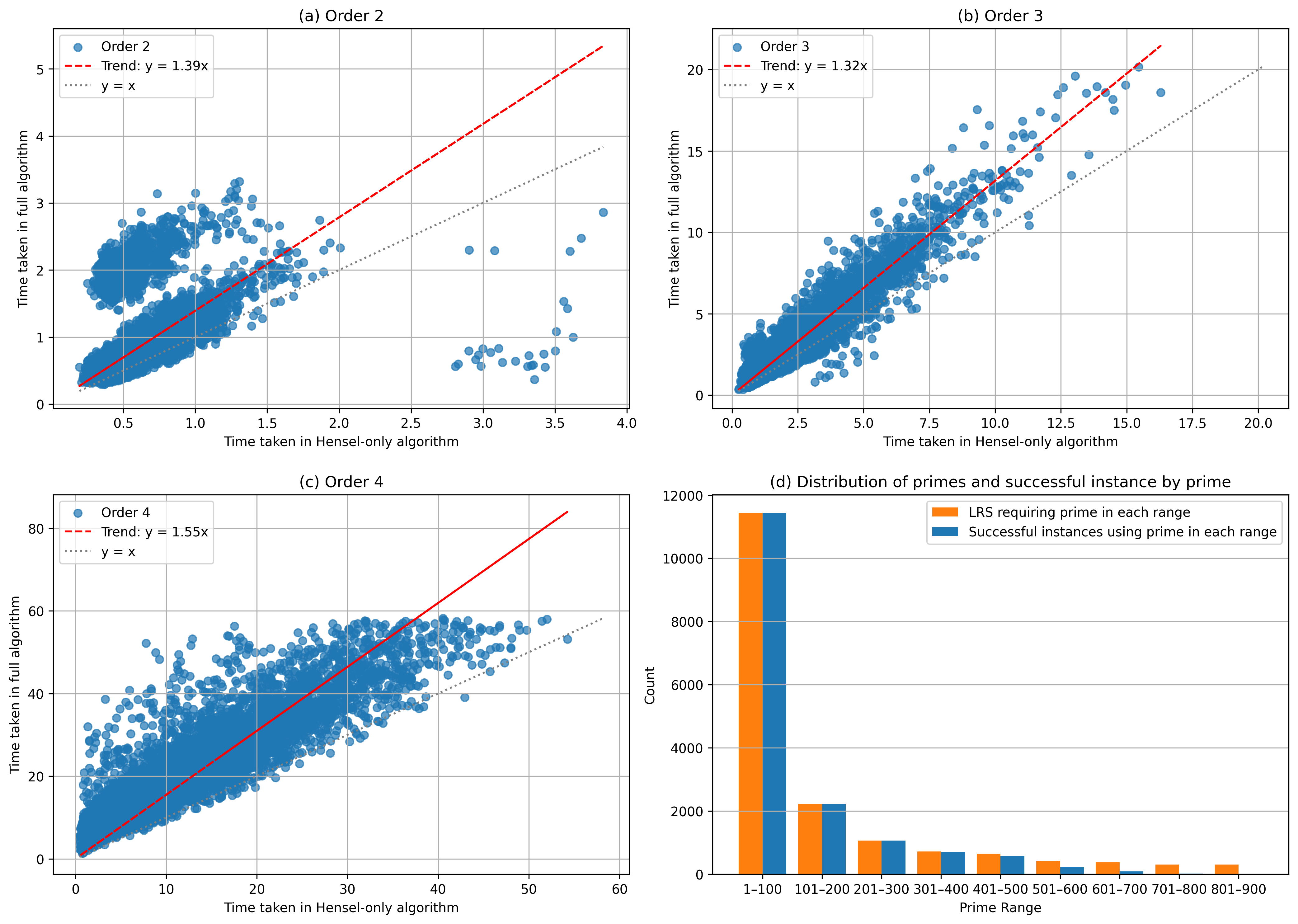}
    \caption{(a)-(c): Timing comparison between the Hensel-only and full algorithm (in seconds) in cases where both succeed within 60s. 
    (d): Distribution of instances requiring a prime in a given range, and the instances for which the Hensel-only algorithm terminates in 60s (displayed up to 900); the system is highly effective up to primes of around 400-500 in this time limit.}
    \label{fig:comparisons}
\end{figure}

In the case where the characteristic polynomial of $\LRS{u}$ splits in $\Z_p$, it is straightforward to implement Algorithm \ref{alg:fullAlg}. We have written an implementation that we have integrated into the \textsc{Skolem} Tool, first developed as part of~\cite{SkolemMeetsSchanuel}. The source code is available at~\cite{bacik_2025_16794130}. The tool is written in Python, using the SageMath computer-algebra system.

Concretely, the tool first searches for the smallest prime greater than some prescribed lower bound (by default 3) such that $g$ splits in $\Z_p$, before computing the $p$-adic zeros of $\LRS{u}$. We have also implemented a shortened ``Hensel-only'' algorithm that assumes all zeros have multiplicity 1 and uses only Hensel's Lemma to identify zeros. This avoids having to compute a multiplicatively independent generating set for the characteristic roots, having to consider irreducible factors, and having to carry out a parallel search for rational zeros. LRS with $p$-adic zeros of multiplicity greater than 1 are relatively rare, so this shortened algorithm terminates for most (but not all) LRS in general\@.

When a zero is found, the tool attempts to determine if it is a (twisted) integer or rational zero by guessing reasonably-sized rational numbers that match the p-adic expansion of the zero found to certain precision. If this procedure is inconclusive the unidentified zero is likely to be transcendental over $\Q$, but we cannot be sure, thus it is reported as unknown. In the interface, the user can request the $p$-adic expansion up to any given precision.


We report the analysis in \cref{fig:comparisons} and \cref{table:experiments1,table:experiment2,table:primes}. Our experiments considered the same randomly generated set of LRS as used in \cite{SkolemMeetsSchanuel}, and ran with up to 24 instances in parallel on 32 core (including Hyper-Threading) Intel Xeon E5-2667 v2 machines with 256GB RAM. Instances which are degenerate or identically zero are excluded as unsupported. 

\cref{table:experiments1} considers the shortened ``Hensel-only'' algorithm. With a sixty-second timeout the tool completes almost all instances up to order 4, but reduces to around half of instances at order 5.  In \cref{table:primes} and \cref{fig:comparisons}d we considered the prime required on a 50\% sample of the dataset, without computing the zeros (in order to speed up the computation). The tool appears to succeed within sixty seconds for LRS which require a prime up to between 400 and 500 (see  \cref{fig:comparisons}d). We observe the growth is approximately $d!$, and provides evidence of why the tool can easily handle order $4$ and most of order $5$, but order 6 would be an order of magnitude slower. The high standard deviation show that there is quite some spread.

The performance of \cref{alg:fullAlg} is depicted in \cref{table:experiment2}. With a 60s timeout the tool is effective up to order $4$, but a longer timeout would be required at order $5$. Counts of zeros, zero types and average prime are not shown, as the result is the same as the ``Hensel-only'' algorithm reported on in Table \cref{table:experiments1}. Not a single example of an LRS with a multiplicity of order greater than 1 was found in this randomly generated set of LRS, showing that these instances are relatively rare. Where both succeed for orders 2-4, the full algorithm is between 1.3 to 1.5 times slower than the Hensel-only approach (see~\cref{fig:comparisons}a-c). 

\bibliography{main.bib}

\end{document}